\documentclass{article}

\usepackage{epsfig,graphicx,float}
\usepackage{amsmath,amssymb}
\usepackage{authblk}

\newcommand{\be}{\begin{equation}}
\newcommand{\ee}{\end{equation}}
\newtheorem{theorem}{Theorem}[section]
\newtheorem{definition}[theorem]{Definition}

\newenvironment{proof}{
\noindent {\it Proof.}}{\hfill$\Box$ }

\title{Localized eigenvectors on metric graphs} 

\author[1]{\normalsize{H. Kravitz} \thanks{hkravitz@pdx.edu}}
\author[2]{\normalsize{M. Brio}\thanks{brio@math.arizona.edu}}
\author[3]{\normalsize{J.-G. Caputo }\thanks{jean-guy.caputo@insa-rouen.fr}}

\affil[1]{Fariborz Maseeh Department of Mathematics and Statistics, Portland
State University, 1825 SW Broadway, Portland, OR 97201, United States
of America}
\affil[2]{Department of Mathematics, University of Arizona,
617 North Santa Rita Avenue, Tucson Arizona, 85721, United States of America}
\affil[3]{Laboratoire de Math\'ematiques, INSA de Rouen Normandie,
Avenue de l'Universit\'e , Saint-Etienne du Rouvray, 76801, France}

\date{\ }

\begin{document}

\maketitle

\begin{abstract}
Using our previously published algorithm, we analyze the eigenvectors 
of the generalized Laplacian for 
two metric graphs occurring in practical applications.
As expected, localization of an eigenvector is rare and 
the network should be tuned to observe exactly localized eigenvectors.
We derive the resonance conditions to obtain localized eigenvectors
for various geometric configurations and their combinations to form more 
complicated resonant structures. These localized eigenvectors suggest 
a new localization indicator based on the $L_2$ norm.
They also can be excited, even with leaky boundary conditions, as shown by
the numerical solution of the time-dependent wave equation on the
metric graph.
Finally, the study suggests practical ways to make resonating systems based
on metric graphs.  
\end{abstract}

\section{Introduction} 

Partial differential equations (PDEs) on networks arise in many physical
applications such as gas and water networks  \cite{Herty10,water},
 electromechanical waves in
a transmission grid \cite{Kundur94}, air traffic control \cite{wb08},
microwave networks \cite{microwave}
and random nanofiber lasers \cite{Gaio,Cipolato}
to name a few. In many cases, the problem can be linearized. Then separation
of variables yields a Helmholtz (or Schr\"odinger) problem on the network.
The associated eigenvalues and eigenvectors play a fundamental role
in the theoretical analysis of the networks \cite{bk13}
and the practical applications.

The underlying mathematical model consists of a metric graph:
a finite set of
vertices connected by arcs (oriented edges) on which a metric is
assigned. Different coupling conditions can be implemented at the vertices.
The simplest assumes continuity of the field and zero
total gradient at the vertices (Kirchhoff's law). The standard
one-dimensional Laplacian together with these boundary conditions results
in a generalized Laplacian and associated Helmholtz eigenvalue problem.
It can be shown that with these coupling conditions (continuity and
Kirchhoff's law) the operator is self-adjoint \cite{bk13}, yielding
real eigenvalues and orthogonal eigenvectors. The eigenvectors form a
complete basis of the appropriate set of square integrable functions
on the graph. This spectral framework plays a key role in linear
PDEs. We review and apply it in the present article.

Combinatorial or discrete graphs bear some similarities to metric graphs.
For these graphs, edges only describe connections between vertices.
For undirected graphs, both the adjacency matrix and the graph Laplacian
are symmetric---they have real eigenvalues, and the
eigenvectors can be chosen orthogonal. Some eigenvectors have non zero
components on just a few vertices, they are {\it localized}
and this localization affects transport properties \cite{cks13}.
In a recent article \cite{cckp20} we showed that
the Laplacian eigenvectors of chains connected to complete graphs are
localized in the complete graph regions where connectivity is high.
These results were also found in the systematic study by
Hata and Nakao \cite{Hata} where they analyzed
graphs with random degree distributions.
The localization of eigenvectors found here is purely
topological because the Laplacian has equal weights.
Random weights introduce additional possibilities for localizing
the eigenvectors, see the pioneering 1958 study by Anderson \cite{Anderson}
of a Schr\"odinger matrix equation with random diagonal and 
off-diagonal elements and the very large literature that followed it. 

In contrast, the localization of eigenvectors on arcs of random metric quantum 
graphs has only been studied in a handful of articles and many of these
assume non standard interface conditions at the vertices. 
For the standard conditions, Schanz and Kottos \cite{sk03} established conditions
for localized eigenvectors in polygons. They use a scattering theory formalism
because they were interested in problems of quantum chaos. 
In another interesting study, Gnutzman, Schanz and Smilansky \cite{gss13}, using similar tools, established that localized eigenvectors cannot 
exist in trees and recovered
the conditions for localization in polygons. Because of the formalism 
used and despite their importance, these 
results are not so well known in the engineering community studying networks.

To study these practical network problems, we introduced 
recently a systematic procedure to compute eigenvalues
and eigenvectors of arbitrary order for general metric graphs \cite{method}. 
In the present article, we use this method
to analyze two metric graphs occurring in engineering applications and examine
the occurrence of localized eigenvectors. Using simple arguments, we examined 
the conditions to observe exactly localized eigenvectors, starting from
the simplest. We found the precise resonance conditions on the
lengths and the eigenvalues for several geometric
configurations such as a cycle with two edges and a polygon; we recover some results of \cite{sk03}, \cite{gss13}
and find new ones. In particular, we show how these localized eigenvectors 
can be connected to form a larger localized eigenvector; this can be 
important for resonator applications.
These results on localized eigenvectors prompted us to define 
a localization criterion giving the number of edges
involved in a localized eigenvector. Finally, the numerical solution 
of the wave equation on the metric graph reveals how these localized 
eigenvectors can be excited from a broadband initial condition. \\
The article is organized as follows. The statement of the problem,
A brief review of the spectral properties
of metric quantum graphs and of our computational algorithm are given in Section 2.
There we also compute numerically and characterize the eigenvectors of two metric graphs.
Section 3 lists exact resonant conditions to obtain localized eigenvectors in various geometric
configurations. The localization criteria are discussed in Section 4. We also show how the solution
of the wave equation with a leaky boundary converges to a localized eigenvector.
Section 5 concludes the article.

\section{Spectral theory for metric graphs }

We first recall the spectral theory formalism for completeness.
We will then illustrate it on two examples from engineering applications.

Consider a finite metric graph with $n$ vertices
connected by $m$ arcs (oriented edges) of length $l_j, \; j=1: m.$
Each edge is parameterized by its length $x$ 
from the origin vertex $x=0$ to the terminal vertex $x=l_j$. 
We recall the definition of the {\it degree} of a vertex: 
the number of edges connected to it.

On this graph, we define the vector component wave equation
\be\label{vwave} U_{tt} -{\tilde \Delta } U=0,\ee
where $U \equiv (u_1,u_2, \dots,u_m)^T$. Each component
satisfies the one-dimensional wave equation inside the respective arc,  
\be\label{wave}
{u_j}_{tt} -{u_j}_{xx} =0, ~~~  j=1, 2, \dots ,m \ee
In addition, the solution should be continuous at the vertices and also satisfy the 
Kirchhoff flux conditions at each vertex of degree $d$
\be\label{kirchof}
\sum_{j=1}^{d} \; {u_j}_{x} =0,
 \ee
where $\displaystyle {u_j}_{x}$ represent the outgoing fluxes for arc $j$ emanating from the vertex.

Consider equation (\ref{vwave}). Since the problem is linear, we can separate time and space
and assume a harmonic solution $U(x,t)= e^{ikt} \: V(x)$.
We then get a Helmholtz or Schr\"{o}dinger eigenproblem for $V$ on the graph
\be\label{helm}
-{\tilde \Delta } V = k^2 V ,\ee
and where ${\tilde \Delta }$ is the generalized Laplacian, i.e. the
standard Laplacian on the arcs together with the coupling conditions at
the vertices.
Note that we exclude all degree two vertices since due to
continuity and the Kirchhoff condition, two edges sharing
such a vertex can be merged into a single edge---see \cite{berkolaiko17}.

The generalized eigenvalue problem (\ref{helm}) admits an inner product
obtained from the standard inner product on $L_2$ space---see \cite{bk13}. We have
\be\label{inproduct}
\langle f ,g \rangle \equiv \sum_{arc ~~j} \langle f_j ,g_j \rangle,~~~~ \langle f_j ,g_j \rangle=\int_{0}^{l_j} f_j(x) g_j(x) dx . \ee
A  solution in terms of Fourier harmonics on each branch $j$ of length $l_j$ is 
\be\label{simple}
v_j(x) = A_j \sin k x + B_j \cos k x . \ee
Writing down the coupling conditions at each vertex, one obtains a 
homogeneous linear system whose $k$-dependent matrix is 
singular at the eigenvalues.

Using solution (\ref{simple}) on each arc with unknown coefficients $A_j$
and $B_j$, the coupling conditions at each vertex yield 
the homogeneous system
\be\label{Meqn}
M(k) X =0, \ee
of $2m$ equations for the vector of $2m$ unknown arc amplitudes  
$$\displaystyle X=(A_1,B_1, A_2, B_2, \dots, A_m, B_m)^T  . $$ 
The matrix $M(k)$ is singular at the eigenvalues $-k^2$. We call
these $k$-values resonant frequencies.
A practical and robust computational algorithm for the computation of these
eigenvalues and eigenvectors was proposed and studied in \cite{method}.

For each resonant frequency $k_q$, the eigenvectors $V^q$ 
span the null space of the matrix $M(k_q)$. They can then be written as
\be\label{Vi}
V^q = 
\begin{pmatrix} A^q_1 \sin k_q x + B^q_1 \cos k_q x \cr 
                A^q_2 \sin k_q x + B^q_2 \cos k_q x \cr
                \dots  \cr
                A^q_m \sin k_q x + B^q_m \cos k_q x \cr
\end{pmatrix}
\ee
They can be normalized using the scalar product defined above. We have
\be\label{vivi}\lVert V^q \rVert^2 =\langle V^q,V^q \rangle = \sum_{j=1}^m \langle V^q_j,V^q_j \rangle, 
\ee
where 
\be\label{vqj} 
V^q_j = A^q_j \sin k_q x + B^q_j \cos k_q x, \ee
 and $\langle V^q_j,V^q_j \rangle$
is the standard scalar product on $L_2([0,l_j])$ . This defines
a broken L$_2$ norm or graph norm. 
The scalar product $\langle V^q_j,V^q_j \rangle$ can be computed explicitly 
$$\langle V^q_j,V^q_j \rangle= 
\left ( {A^q_j}^2 + {B^q_j}^2 \right ) {l_j \over 2} $$
\be\label{vpvp}
+ {\sin {2 k_q l_j} \over 4 k_q} \left ( -{A^q_j}^2 + {B^q_j}^2 \right )
+ A^q_j B^q_j { 1 - \cos {2 k_q l_j} \over 2 k_q}.\ee
It has been shown that, for the standard coupling conditions used here
(continuity and Kirchhoff's condition), the eigenvectors $V^i$ form a complete 
orthogonal basis of the Cartesian product 
$L_2([0,l_1]) \times L_2([0,l_2])  \dots \times L_2([0,l_m]) $---see \cite{bk13}.

Once the eigenvalue problem is solved, 
one can proceed with the spectral solution of 
the time-dependent problem, exactly as for the one-dimensional wave equation
on an interval. 
For that, expand the solution of the wave equation on the
graph (\ref{vwave}) using the eigenvectors, 
\be\label{spec1}
U(x,t) = \sum_{q=1}^\infty a_q(t) V^q ,\ee
and obtain a simplified description of the dynamics in terms of the amplitudes
$a_q$.

\subsection{ Localized eigenvectors : two numerical examples}

Like the eigenvectors of the discrete Laplacian, eigenvectors
of the generalized Laplacian on a network can be localized
in the following sense.
\begin{definition}[{\rm Localized eigenvector}]
\label{localevector}
An eigenvector $V^i$ of the generalized Laplacian operator
with the standard coupling conditions is localized if its
components $V^i_j$ (see (\ref{vqj})) satisfy
$V^i_j \neq 0$ for a finite number of edges $j$ and $V^i_j=0$
for the rest.
\end{definition}
Such a localized eigenvector plays an important role in the 
dynamics of the wave equation.

We consider two graphs with no symmetries and arbitrary edges to illustrate how frequently 
localized eigenvectors appear. To identify these eigenvectors, we 
compute the $L_2$ norm ratio 
\be \label{ej} e_q(j) \equiv {\langle V^q_j,V^q_j \rangle \over \langle V^q,V^q \rangle} \ee 
for each edge $j=1,2,\dots ,m$.

\subsection{Graph G14}

We introduce the 14 edge graph---see Fig. \ref{g14}. It was adapted from IEEE case 14 \cite{IEEE} by
eliminating the degree two vertices. Such a metric graph can be used to model how 
electromechanical waves propagate in an electrical grid \cite{Kundur94}.
A localized eigenvector in this context would correspond to an accumulation of energy on just a few
equipments and this could cause their failure.
\begin{figure}[H]
\centerline{
\epsfig{file=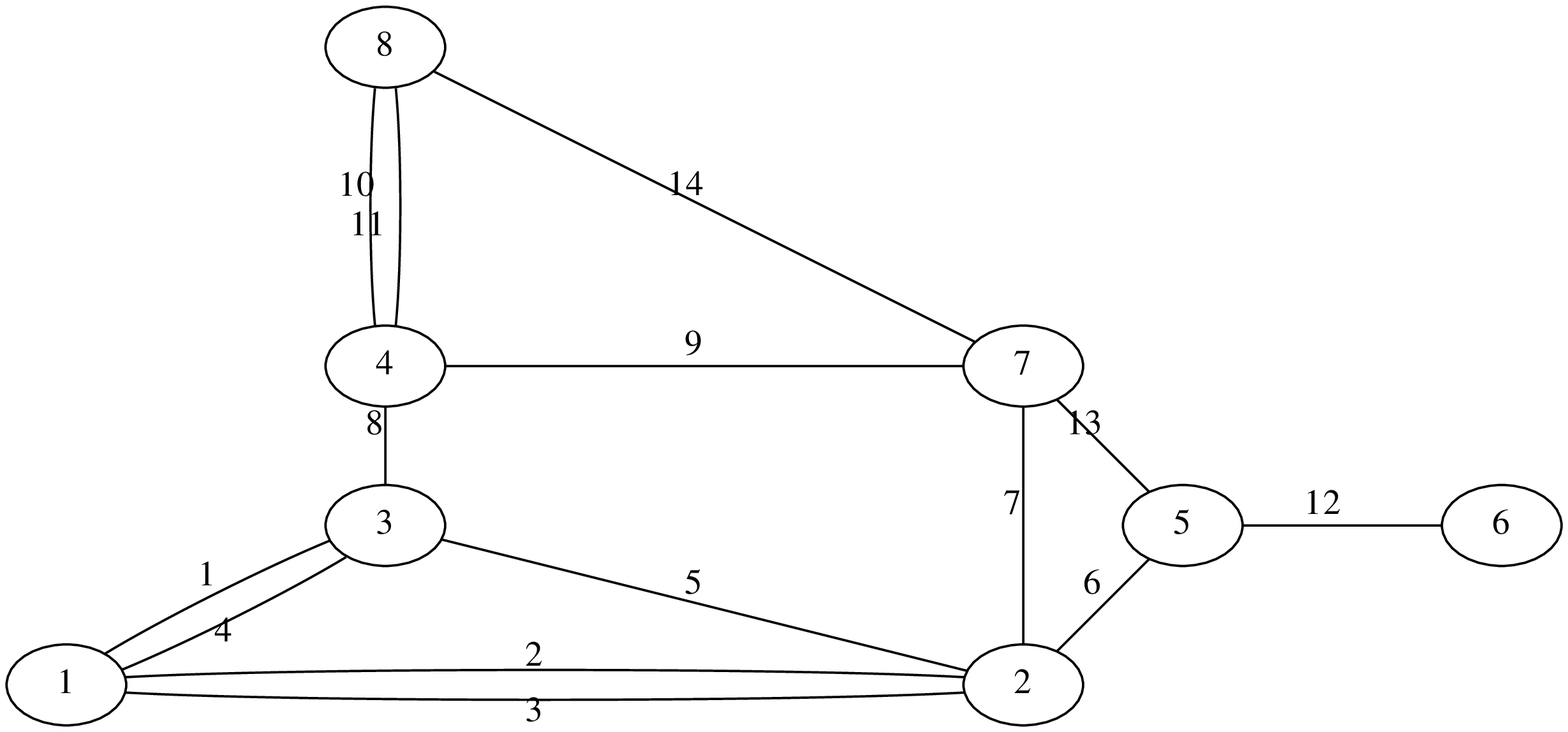,height=6cm,width=12.675 cm,angle=0}
}
\caption{The 14 vertex metric graph G14} 
\label{g14}
\end{figure}
The lengths $l_i, i=1,\dots,m=14$ are given in Table \ref{tab2}.
\begin{table} [H]
\centering
\begin{tabular}{|l|c|c|c|c|c|r|}
\hline
$l_1$       & $l_2$      & $l_3$ & $l_4$ & $l_5$       &  $l_6$     & $l_7$\\     
11.91371443 & 7.08276253 &   6   & 2.236067977 &4.123105626 & 1.414213562 &  2\\
            &            &       &   & & &            \\ \hline
 $l_8$ &  $l_9$      & $l_{10}$  & $l_{11}$&  $l_{12}$  &$l_{13}$   & $l_{14}$ \\
1 & 4.7169892  & 4.472135955& 2       & 2  &1.414213562 & 4.472135955 \\ \hline 
\end{tabular}

\caption{The lengths $l_i$ for the graph G14 . }
\label{tab2}
\end{table}

\begin{figure}[H]
\centerline{
\epsfig{file=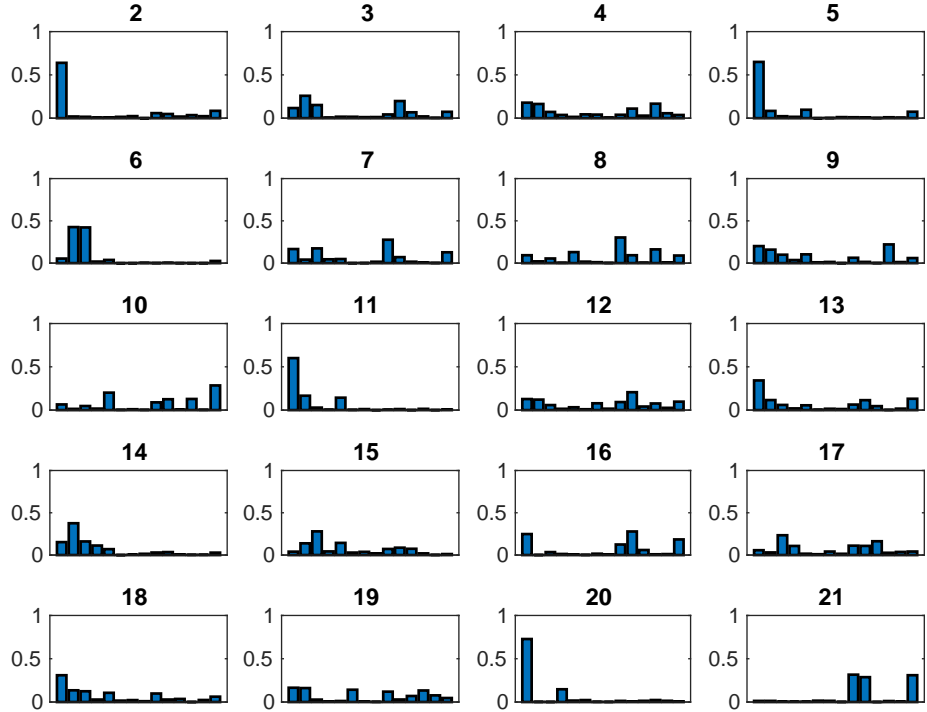,height=11.25 cm,width=15 cm,angle=0}
}
\caption{Histograms of the $L_2$ norm ratio $e_q(j)$ (\ref{ej})  for
$q=2,\dots 21$ from left to right and top to bottom for the G14 graph.}
\label{lhg14}
\end{figure}
Fig. \ref{lhg14} presents the histograms of the $L_2$ norm ratio $e_q(j)$ 
for eigenvectors $V^q, ~~q=2,\dots 21$ for the G14 graph. Note that
$q=2,5,6, 11, 20$ and 21 correspond to eigenvectors where 
$e_q(j) \ge 0.5$ for one or two edges $j$ and $e_q(j) < 0.05$
for the other edges. 
In Fig. \ref{lg14} we present the approximately localized eigenvectors corresponding
to $q=2,5,6, 11, 20$ and 21.
For a given eigenvector $V^q$, we present for each edge $j$ the quantity 
$e_q(j)$.
\begin{figure}[H]
\centerline{
\epsfig{file=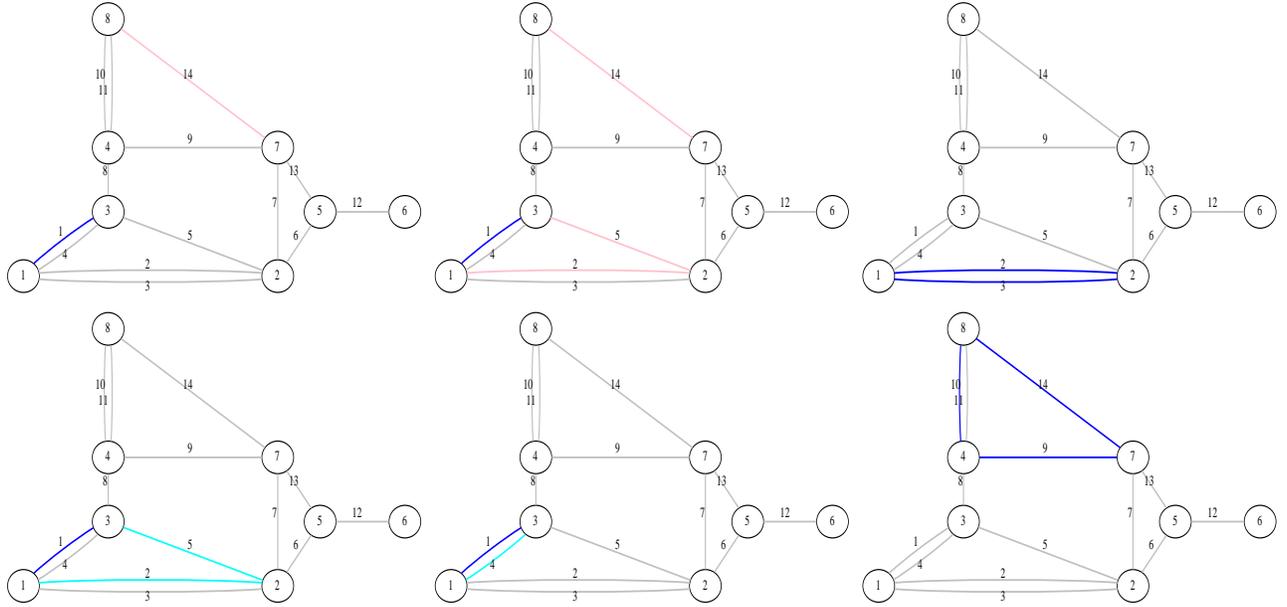,height=8 cm,width=16.9 cm,angle=0}
}
\caption{Approximately localized eigenvectors $V^q$ for $q=2,5,6, 11, 20$ and 21.
We present for each arc $j$ the quantity $e_q(j)$
with the color code $ e_q(j) < 0.06$ (grey), $0.06 < e_q(j)< 0.12$ (pink),
$0.12 < e_q(j) < 0.2$ (cyan) and $0.2 < e_q(j)$ (blue). 
The values of $k_q$ are given in the table below. }
\label{lg14}
\end{figure}

The values of $k_q$ are presented in Table \ref{tab3} below.
\begin{table} [H]
\centering
\begin{tabular}{|l|c|r|}
\hline
0.2347645148 & 0.4657835674  & 0.480197067 \\ \hline
0.8078723081 & 1.3322287766 &  1.379308786 \\ \hline
\end{tabular}
\caption{The values of $k_q$ shown in Fig. \ref{lg14}.}
\label{tab3}
\end{table}

\subsection{Buffon graph }

The second example we present comes from a study by Gaio et al \cite{Gaio}
suggesting that lasers can be produced by fusing randomly placed 
nanometric optical waveguides.
The resulting graph appears as a series of scattered needles.
Such a Buffon's needle graph with 165 arcs and 104 vertices is
shown in Fig. \ref{buf1}.

In Gaio's study, the fibers are active so that they would amplify the
field. The lasing effect would come from a balance between this amplification
and damping. In addition, there would be transparent conditions at the
boundaries so that any out-going radiation would be lost.
Then, a laser effect would occur on the localized eigenvectors of the
Laplacian and only on those because the extended eigenvectors would be damped
due to the boundary conditions.
\begin{figure}[H]
\centerline{
\epsfig{file=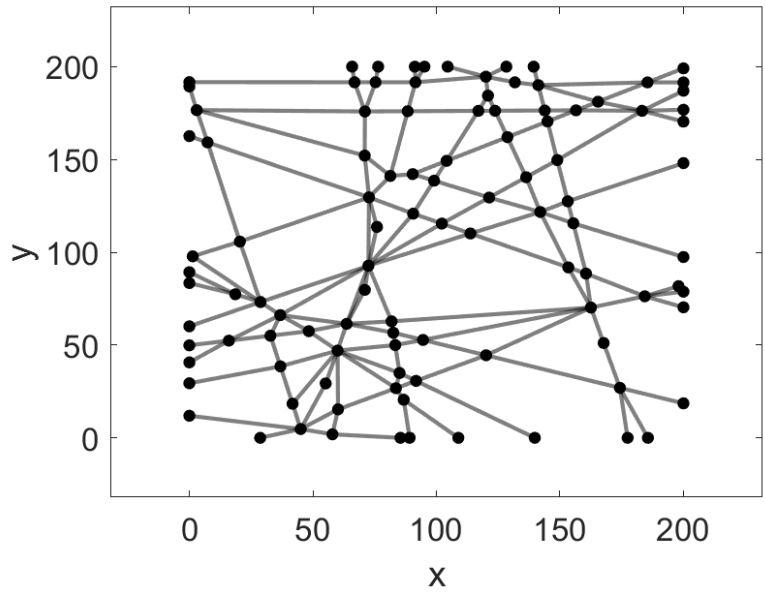,height= 8.4675 cm,width=11.28 cm,angle=0}
}
\caption{A Buffon's needle graph with $m=104$ vertices
and $m=165$ arcs.}
\label{buf1}
\end{figure}

Fig. \ref{lhbuf} presents the histograms of the $L_2$ norm ratio $e_q(j)$
for eigenvectors $V^q, ~~q=2,\dots 21$ for the Buffon graph of 
Fig. \ref{buf1}.
\begin{figure}[H]
\centerline{
\epsfig{file=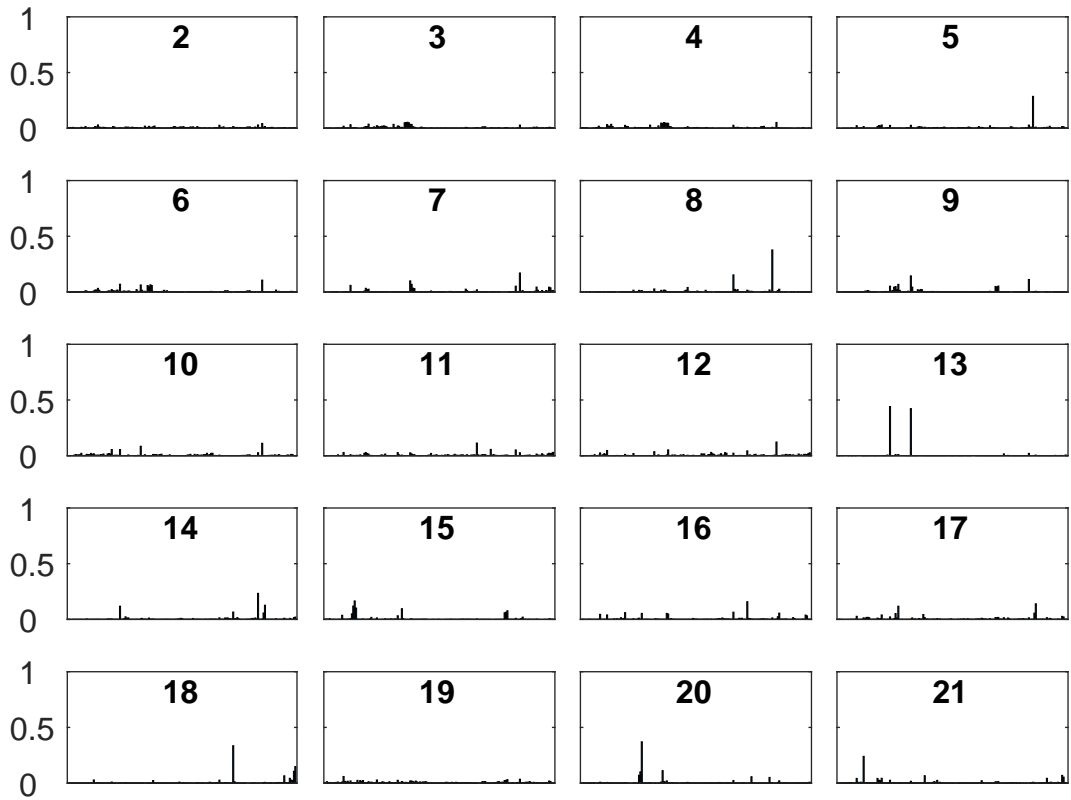,height=10.536cm,width= 14.107 cm,angle=0}
}
\caption{Histograms of the $L_2$ norm ratio $e_q(j)$ (\ref{ej})  for
$q=2,\dots 21$ from left to right and top to bottom for the Buffon graph.}
\label{lhbuf}
\end{figure}
The vectors $V^q$ for $q=5,8,13,18,20$ and $21$ are approximately localized. 
They are plotted below.
\begin{figure}[H]
\centerline{
\epsfig{file=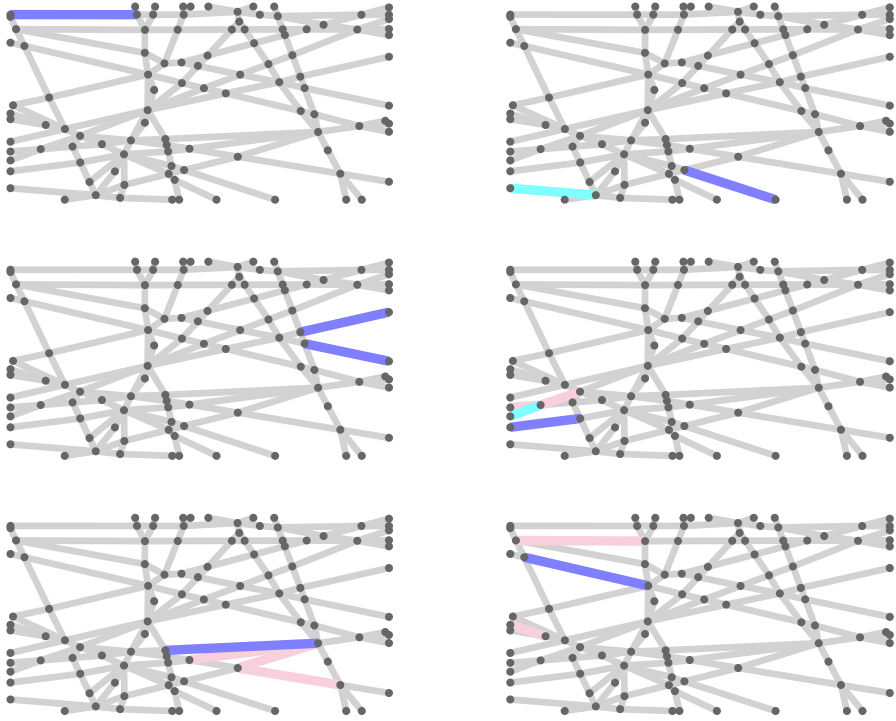,height=10.25 cm,width=13.357 cm,angle=0}
}
\caption{Plots of the localized eigenvectors $V^q$ for $q=5,8,13,18,20$ and $21$.
The color code is $ e_q(j) < 0.06$ (grey), $0.06 < e_q(j)< 0.12$ (pink),
$0.12 < e_q(j) < 0.2$ (cyan) and $0.2 < e_q(j)$ (blue).}
\label{locbuf}
\end{figure}
The corresponding values of $k_q$ are shown in Table \ref{tab5}
\begin{table} [H]
\centering
\begin{tabular}{|l|c|r|}
\hline
0.016866927 & 0.022719737 & 0.029337666 \\ \hline
0.034764362 & 0.036236231 & 0.036802227 \\ \hline
\end{tabular}
\caption{The values of $k_q$ shown in Fig. \ref{locbuf}.}
\label{tab5}
\end{table}

The results of this section illustrate that
for arbitrary metric graphs, approximately localized eigenvectors occur,
in particular we observed them for large $k$.
However, this localization is not exact. This is a known result, 
see the statement "there are no perfect scars for generic graphs" by
Schanz and Kotos \cite{sk03}. For the laser application \cite{Gaio},
a lasing phenomenon is therefore unlikely to appear in 
a random arrangement of nanometric waveguides.

For exactly localized eigenvectors to exist, we need a precise arrangement 
of the lengths of the arcs involved. We give these resonance conditions 
in the next section.

\section{Exactly localized eigenvectors}

We examine configurations
that lead to {\it exactly localized eigenvectors} 
following the definition (\ref{localevector}). In the rest of
this section we drop the adjective {\it exactly}.
We find that localized
eigenvectors exist on two connected leaves, as a 1-2 state in a pumpkin,
as a triangle 1-2-3 and a quadrilateral 1-2-3-4---see Table \ref{texact}.
The analysis also enables us to rule out single arc, leaf, two
connected arcs and degree three vertex localized eigenvectors---see Table \ref{tnoexact}.

We show the computations in detail for the 1-2-3 triangle.
Calculations for the other examples are given in the appendix.

\subsection{A localized eigenvector, the Triangle 1-2-3}

Consider the configuration of Fig. \ref{tri123} where a triangle 
of edges $l_1,l_2,l_3$ is embedded in a graph. 
\begin{figure}[H]
\epsfig{file=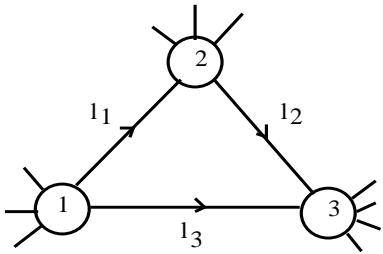,height=5 cm,width=12 cm,angle=0}
\caption{A triangle 1-2-3 embedded in a graph.}
\label{tri123}
\end{figure}
We have the following theorem.
\begin{theorem}
\label{tria123}
A localized eigenvector exists on a a triangle 
of edges $l_1,l_2,l_3$ embedded in a graph if 
there exists three integers $n_1,n_2,n_3$ such that
$$ {l_1 \over n_1}= {l_2 \over n_2}={l_3 \over n_3}$$ 
and $n_1 + n_2 + n_3$ is even.\\
The eigenvalue is $-({n_1 \pi \over l_1})^2$ and the eigenvector is
$$V= \sin k x (1 , (-1)^{n_1} , -1)^T $$
\end{theorem} 

\begin{proof}
On each edge $j$, we have $V_j = A_j \sin kx + B_j \sin kx$.
To have a localized eigenvector, we need that at each vertex $V=0$
and to balance the fluxes. These conditions are
\begin{eqnarray*}
V_1(0)=V_3(0)=0, \\ 
{V_1}_x(0)+ {V_3}_x(0)=0, \\
V_1(l_1)=V_2(0)=0, \\ 
{V_1}_x(l_1)- {V_2}_x(0)=0, \\
V_3(l_3)=V_2(l_2)=0, \\ 
{V_3}_x(l_3)+ {V_2}_x(l_2)=0,
\end{eqnarray*}
yielding
\begin{eqnarray*}
k l_1 = n_1 \pi, ~~k l_2 = n_2 \pi,~~k l_3 = n_3 \pi, \\
{l_1 \over l_2} = {n_1 \over n_2},~~{l_1 \over l_3} = {n_1 \over n_3} , \\
B_1 = B_2 = B_3 =0, \\
A_3 = -A_1 , \\
A_2 = A_1 c_1 , \\
A_2 c_2 + A_3 c_3 =0, 
\end{eqnarray*}
where $n_1,n_2,n_3$ are integers and $c_1 = \cos k l_1, \dots$.
Using the last relation, we get the condition
\be \label{tricon} (-1)^{n_1+n_2}= (-1)^{n_3},\ee
so that $n_1+n_2 +n_3$ is even. To summarize, we have
a triangle eigenvector if there exists four integers $n_0,n_1,n_2,n_3$ such that
\begin{eqnarray}
{l_1 \over n_1}= {l_2 \over n_2}={l_3 \over n_3} , \\
n_1+n_2 +n_3 = 2 n_0 .
\end{eqnarray}
The eigenvector is
$$V= \sin k x (1 , (-1)^{n_1} , -1)^T $$
\end{proof}

Fig. \ref{trir} shows such a state for $k=1$ in the G14 graph, 
with $l_6  = 2 \pi, ~l_7= 3\pi$ and $ l_{13} = 7\pi$.
\begin{figure}[H]
\centerline{
\epsfig{file=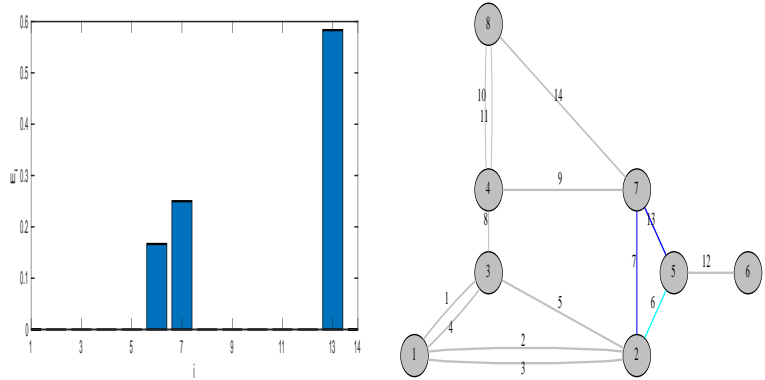, height= 10 cm, width = 5 cm, angle=90}
}
\caption{A triangle 1-2-3 "embedded" in a graph.  The active arcs are 
6,7 and 13. }
\label{trir}
\end{figure}
Here, we recover in a simple way the conditions obtained by Gnutzmann et 
al \cite{gss13} using scattering theory arguments. Similarly, we can
obtain a localized eigenvector on a quadrilateral or any polygon.
Fig. \ref{quad1234} shows such an exactly localized eigenvector
on the quadrilateral 5-7-8-9 for the G14 graph where we chose
$$l_5  = 2\pi, ~l_7  = 3\pi,~ l_8  = 5\pi,~ l_9  = 6\pi .$$
\begin{figure}[H]
\centerline{
\epsfig{file= 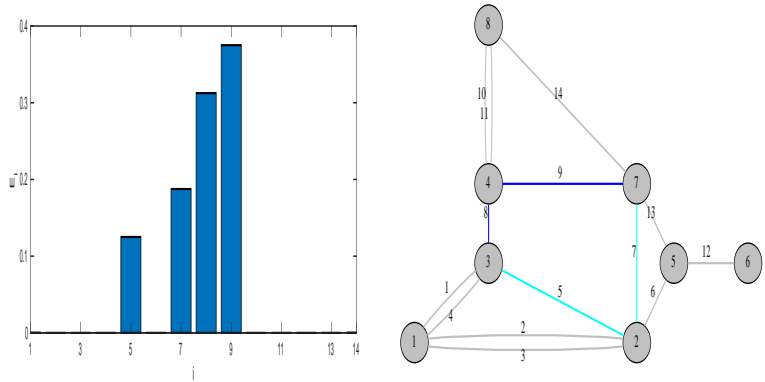, height= 10 cm, width = 5 cm, angle=90}
}
\caption{A quadrilateral 1-2-3-4 embedded in a graph. The
active arcs are 5,7,8 and 9. }
\label{quad1234}
\end{figure}

\subsection{1-2 localized eigenvector on a pumpkin subgraph}

We present here a localized eigenvector that is new to the best of our
knowledge. It is a 1-2 or more resonance in a pumpkin subgraph.
Pumpkin graphs were studied in detail by Berkolaiko \cite{berkolaiko2} 
who introduced this terminology.
\begin{definition}
An $m$-pumpkin graph consists
of two vertices and $m$ parallel edges of possibly different lengths running
between them.
\end{definition}
\begin{figure}[H]
\centerline{
\epsfig{file=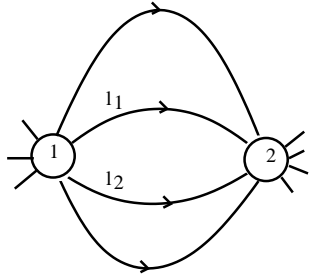,height=4.14 cm,width=9.61 cm,angle=0}
}
\caption{A 2-pumpkin 1-2 embedded in a graph.}
\label{pum12}
\end{figure}
Consider the configuration of Fig. \ref{pum12} where a pumpkin is
embedded in a graph. For certain edge lengths $l_1,l_2$, there exists
a localized eigenvector on these two edges.
\begin{theorem}
\label{pump12}
A localized eigenvector exists on a pumpkin subgraph of a metric graph 
with at least two edges $l_1,l_2$ if there are two integers $n_1,n_2$
of same parity such that
$${l_1 \over n_1}= {l_2 \over n_2}  . $$
\end{theorem} 

\begin{proof}
As for the triangle, we want $V$ to be zero at the vertices and to balance the
flux, then
\begin{eqnarray*}
V_1(0)=V_2(0)=0, \\ 
V_1(l_1)=V_2(l_2)=0, \\ 
{V_1}_x(0)+ {V_2}_x(0)=0, \\
{V_1}_x(l_1)+ {V_2}_x(l_2)=0  .
\end{eqnarray*}
This yields the system of equations
\begin{eqnarray*}
B_1=B_2=0, \\
A_1 \sin k l_1 =A_2 \sin k l_2 =0, \\
A_1+A_2=0, \\
A_1 \cos k l_1  +A_2 \cos k l_2 =0.
\end{eqnarray*}
We then obtain
$$\sin k l_1 =0, ~~\sin k l_2 =0 , $$
so that $$k l_1 = n_1 \pi, ~~k l_2 = n_2 \pi ,$$
where $n_1,n_2$ are integers. Then
$\cos k l_1 = (-1)^{n_1}, ~~~ \cos k l_2= (-1)^{n_2}$.
A non-trivial solution $A_1,A_2$ exists only if $\cos k l_1=\cos k l_2$
so that $n_1$ and $n_2$ have the same parity.
The condition on the lengths is then 
\be \label{cpum12}
{l_1 \over n_1}= {l_2 \over n_2}  , \ee
where $n_1 , n_2$ are integers of same parity.
\end{proof}

To illustrate this localized eigenvector, consider the graph shown in
Fig. \ref{pum3}. 
\begin{figure}[H]
\centerline{
\epsfig{file=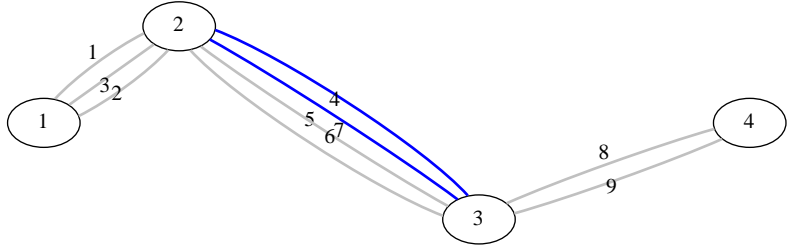,height=3.375 cm,width=10.5 cm,angle=0}
}
\caption{A graph showing a localized 1-2 pumpkin eigenvector.}
\label{pum3}
\end{figure}
The lengths are given in Table \ref{tpum3}, where the arcs 4 and 7 satisfy the
resonance condition (\ref{cpum12}). The localized eigenvector on
arcs 4 and 7 is shown with the color code $ e_q(j) < 0.06$ (grey), $0.06 < e_q(j)< 0.12$ (pink),
$0.12 < e_q(j) < 0.2$ (cyan) and $0.2 < e_q(j)$ (blue).
\begin{table} [H]
\centering
\begin{tabular}{|l|c|c|c|r|}
\hline
$l_1$        &  $l_2$      &  $l_3$       &  $l_4$       &  $l_5$   \\
2.236067977 &  1.414213562 &  1.732050807 &  $\pi$       &  $11~ \pi$  \\
             &             &              &              &   \\ \hline
$l_6$        &  $l_7$      &  $l_8$       &  $l_9$       &    \\
5.167771571  &  9.424777960& 3.605551275  &  5.693156148 &    \\
             &             &              &              &   \\ \hline
\end{tabular}
\caption{The lengths $l_i$ for the graph shown in Fig. \ref{pum3}.}
\label{tpum3}
\end{table}

Similarly, one can have a localized eigenvector on a 3-pumpkin. The derivation
is similar to the 2-pumpkin.
We obtain the conditions 
$$A_1 + A_2 + A_3 = 0, $$
$${l_1 \over n_1}= {l_2 \over n_2}= {l_3 \over n_3}$$ 
where $n_1,~n_2$ and $n_3$ are integers of the same parity. Note 
that the eigenspace has dimension 2.

To conclude this section, Table \ref{texact} gives four 
configurations giving localized eigenvectors. 
\begin{table} [H]
\centering
\begin{tabular}{|c|c|c|c|}
\hline
2 connected & 1-2 eigenvector & triangle 1-2-3&  quadrilateral 1-2-3-4 \\
leaves & in a pumpkin & &   \\\hline
${l_1  \over n_1}={l_2 \over n_2}$ & ${l_1  \over n_1}={l_2 \over n_2}$ & ${l_1  \over n_1}={l_2 \over n_2}={l_3 \over n_3}$  & ${l_1  \over n_1}={l_2 \over n_2}={l_3 \over n_3} ={l_4 \over n_4}$ \\ 
$n_1,n_2$ odd integers & $n_1,n_2$ integers & $n_1,n_2,n_3$ integers  & $n_1,n_2,n_3,n_4$ integers \\ 
 & same parity & $n_1+n_2+n_3$ even & $n_1+n_2+n_3+n_4$ even \\
\hline
\end{tabular}
\caption{Four configurations giving localized eigenvectors and conditions for the
lengths of the arcs.} 
\label{texact}
\end{table}

{\bf Connecting eigenvectors}

Two elementary graphs corresponding to
localized eigenvectors for the same eigenvalue can be connected to yield a 
composite graph for the same eigenvalue. We have the following.
\begin{theorem}
Consider two elementary graphs $G_1,G_2$ corresponding to localized eigenvectors
of the generalized Laplacian for the same eigenvalue. Then, the composite graph
obtained by joining a vertex from $G_1$ to a vertex from $G_2$ has the same
eigenvalue.
\end{theorem}
\begin{proof}

The proof is elementary. The eigenvector components $V^1_j,V^2_k$ are 
zero at each vertex of $G_1$ and $G_2$ respectively so that the 
zero condition is satisfied for both subgraphs.

Since the
fluxes are balanced separately for  $G_1$ and $G_2$, they will be balanced
for the composite graph. This shows that the union of the eigenvectors
$V^1,V^2$ is a localized eigenvector for the composite graph.
\end{proof}

A consequence of this result is that the composite graph can be a subgraph
of a large graph and have the same eigenvalue as long as there are 
no "external" edges i.e. not belonging to $G_1$ and $G_2$.

\subsection{No single arc eigenvector}

The methodology given above also allows us to rule out
geometric situations where no localized eigenvector exists. 
As an example, consider a single arc.
\begin{figure} [H]
\centerline{
\epsfig{file=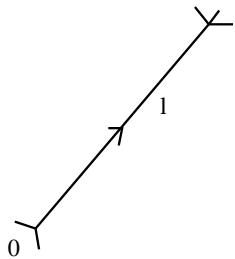,height=4cm,width=4 cm,angle=0}
}
\caption{An arc embedded in a graph. }
\label{arc}
\end{figure}
We have the following theorem
\begin{theorem}
A localized eigenvector of the generalized Laplacian
 cannot exist on an arc embedded in a metric graph.
\end{theorem}
\begin{proof}
To show this, consider Fig. \ref{arc}. The localization
conditions for the solution $$V = A \sin k x + B \cos kx $$
are
$$V(0)=V(l)=0,~~V(0)_x=V_x(l)=0$$
This yields $A=B=0$ so that a localized eigenvector cannot exist.
\end{proof}

This is a known result \cite{sk03}.

Table \ref{tnoexact} gives four configurations where 
localized eigenvectors do not exist. Details of the calculations
are given in the appendix.
\begin{table} [H]
\centering
\begin{tabular}{|c|c|c|c|}
\hline
single & leaf & two connected &  degree 3 vertex \\
arc & eigenvector & arcs &  eigenvector \\
\hline
\end{tabular}
\caption{Four configurations where a localized eigenvector cannot exist.}
\label{tnoexact}
\end{table}

\section{Discussion: localization criterion and excitation}

The histograms of the energy components $\langle V^q_j,V^q_j \rangle$ 
of the localized eigenvectors show differences despite
the fact that the $A$ coefficients are the same. 
This imbalance is due to the different lengths of the arcs $j$ because
$\langle V^q_j,V^q_j \rangle$ scales like $l_j$---see (\ref{vpvp}).

To correct the imbalance, we use equation (\ref{ej}) and 
rescale $\langle V^q_j,V^q_j \rangle$ by $l_j$. For that we
introduce the localization criterion for an eigenvector $V^q$ as
\be\label{loc}
{\cal L}_q  = {\rm max}_j~ {E^q_j} , \ee
where
\be \label{eqj}
{E^q_j} \equiv  { \langle V^q_j,V^q_j \rangle   \over  l_j \sum_k {\langle  V^q_k,V^q_k \rangle  \over  l_k} }  \ee
where we use the graph norm (\ref{vivi}).

The calculations of the previous section for the
two leaf, the triangle and the quadrilateral localized eigenvectors can be
used to exactly compute ${\cal L}_q$. The results are shown in Table
\ref{tabloc}.
To illustrate the usefulness of $E^q_j$, observe that the histograms presented in the left panels of
Figs. \ref{trir} and \ref{quad1234} will have all the same amplitude if $E^q_j$ is used instead of $e_q(j)$ (\ref{ej}); the amplitudes will be 1/3 
and 1/4 respectively, see Table \ref{tabloc}.  
This shows that $1/ {\cal L}_q$ gives the number of active arcs.
The analysis of the previous section shows that an eigenvector is localized on at least two arcs with equal amplitudes
$A_1=A_2$. 
Then a general upper bound is 
$${\cal L}_q \leq 0.5 $$
The quantity $E^q_j$ is the energy density of edge $j$.
When applied to arbitrary metric graphs such as the G14, it indicates
the regions of the graph that are most active.

In contrast, the standard localization criterion 
\be\label{ipr} IPR_q = { \sum_{j=1}^m   \int_0^{l_j} |V^q_j|^4 dx \over 
(\sum_{j=1}^m \int_0^{l_j} |V^q_j|^2 dx )^2} , \ee
used for example by Gaio \cite{Gaio} does not give such precise 
information on the active arcs, as 
it is not based on localized eigenvectors. Table \ref{tabloc}
shows the IPR for the 2-leaf, the triangle and the quadrilateral together
with our estimate ${\cal L}_q$ given by (\ref{loc}). The former gives
the number of active lengths and the latter an estimate of the sum
of the lengths of the edges on which energy is concentrated.
\begin{table} [H]
\centering
\begin{tabular}{|l|c|c|r|}
\hline
                & two leaf    &  triangle    &  quadrilateral   \\  \hline
                &             &              &           \\
${\cal L}_q$    &  1/2        &  1/3         &  1/4      \\ 
                &             &              &           \\
IPR             &  ${3 \over 2(l_1+l_2)}$        &  ${3 \over 2(l_1+l_2+l_3)}$         &  ${3 \over 2(l_1+l_2+l_3+l_4)}$       \\ 
                &             &              &           \\
\hline
\end{tabular}
\caption{Localization criterion ${\cal L}_q$ and Inverse Participation
Ratio (\ref{ipr})  for three localized eigenvectors. }
\label{tabloc}
\end{table}

\subsection{Exciting a localized eigenvector by a broadband pulse}

An important practical issue is how to excite these localized eigenvectors.
For the random fiber laser, the authors of \cite{Gaio,Cipolato}
send an electromagnetic pulse on a region of the network. They
also couple the arcs on the boundary to an outside circuit
to let energy escape. 
Then they expect that the only energy
remaining will be that corresponding to the localized 
eigenvectors.

This argument is correct in principle. To confirm it, consider
the eigenvectors shown in Fig. \ref{lg14}. All localized modes are
away from the vertex 6 so the components at that vertex are
exponentially small.  
Assume a Sommerfeld radiation condition at that vertex,
\be \label{radcon} \epsilon U_t = U_x . \ee
Then the boundary condition there becomes
$$\epsilon {V_j}_x + i k V_j = 0 , $$
which is easily satisfied if ${V_j}_x=V_j=0.$
This simple argument shows that localized eigenvectors
inside the graph
will be preserved when the boundaries of the
network are coupled to a dissipation source.

We illustrate this numerically on the G14 metric graph, using the
finite difference code studied in the article \cite{method}. 
We formed a localized eigenvector on the triangle defined by the
arcs 1,3 and 5 for 
the 18th eigenvalue corresponding to $k=1.133761002$. For that we chose
the lengths of the arcs 1,3 and 5, 
$$l_1 = {4 \pi \over  k} \approx 11.451671,  ~~~
l_3=l_5= {2 \pi \over  k} \approx 5.909775  . $$ 
Vertex 6 has a transparent
boundary condition (\ref{radcon}) ($\epsilon=1$). The other external vertices
have Neuman boundary conditions.
We solve the generalized wave equation using the numerical procedure detailed
in \cite{method} and plot in Fig. \ref{g14wave} two snapshots of the time 
evolution of the components $U_j(x,t)$ on each arc $j$. 
At time $t=0$, $U$ is a gaussian on edge 5 and zero everywhere else, with
1 as initial velocity. The left panel shows
a short time $t=7 ~ 10^4$ and the middle panel a much longer time $t=4 ~ 10^6$.
For the former, the solution is still in a transient state while the 
latter indicates that we reached an asymptotic state corresponding 
to the localized eigenvector on the triangle 1-3-5. There the maximum of
the solution is $2~10^{-2}$ while it is $3 ~10^{-3}$ on the other arcs.
The histogram of
the energies shown in the right panel of Fig. \ref{g14wave} 
shows that the energy is concentrated on the arcs 1-3-5. The energy
of $E_{13}$ is due to the large length of that arc $l_{13}=22$.
\begin{figure}[H]
\centerline{
\epsfig{file= 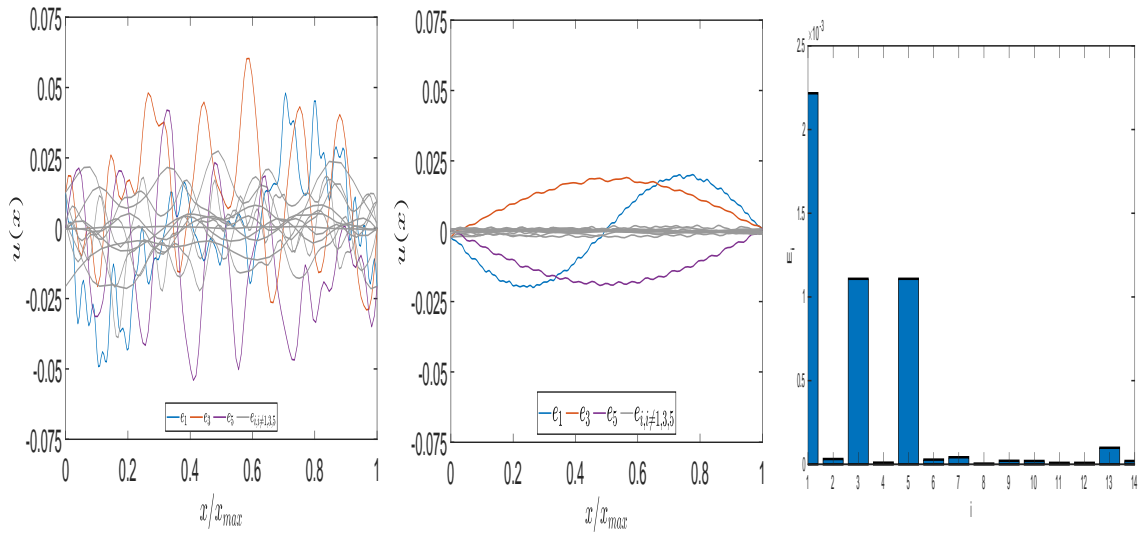, height= 15 cm, width = 7 cm, angle=90}
}
\caption{Numerical solution of the wave equation on the G14 metric graph for
a broadband initial condition and a transparent boundary condition (\ref{radcon})
at vertex 6. The left and middle panels show
snapshots of the solution components $U_j(x,t)$ on each arc $j$ 
for respectively $t=7 ~ 10^4$ and $t=4 ~ 10^6$.
The right panel shows a histogram of the energies $E_j$ on each edge $j$.}
\label{g14wave}
\end{figure}

\section{Conclusion}

In this article we applied our algorithm to study the eigenvectors
of two metric graphs arising in the modeling of the electrical grid and
in a model of a random laser. We find that localized eigenvectors occur
rarely and that the network needs to be tuned specifically for this.

We describe precise resonance conditions on the lengths of the arcs
to obtain exactly localized eigenvectors. Some of these results were
known, new results are the 1-2 arc cycle eigenvector in a pumpkin 
and the connection of localized eigenvectors to form a larger localized
eigenvector.

We define a new localization criterion based on the $L_2$ norm which gives
the number of active edges in an eigenvector, this quantity cannot be
obtained from the standard IPR criterion. An important question is how
to excite these localized eigenvectors?
To answer this we showed, using the time dependent wave equation with 
a leaky boundary, that
a localized eigenvector gets naturally excited in the long term 
from a broadband initial condition.

For the electrical grid application, even approximately localized eigenvectors
can damage equipment so it should be reinforced in the regions
of maximal amplitude of these localized eigenvectors. 
For a laser, a random arrangement of wave guides graph will not in general 
lead to a lasing phenomenon. A laser should be built by associating 
structures corresponding to localized eigenvectors, not random links.

In the future, we plan to extend this study to nonlinear effects and 
examine the stability of the resonance to perturbations (quality factor).
We will also examine non-Kirchhoff coupling conditions.

\section*{Acknowledgements}
JGC thanks the Mathematics Department at the 
University of Arizona for its hospitality during the spring semesters 2022
and 2023. 
He is grateful to the Gaspard Monge foundation for support. HK 
thanks the ARCS Foundation for support.

\appendix

\section{Exactly localized eigenvectors}

\subsection{Two connected leafs}

Two connected leaves form the structure shown in Fig. \ref{tleaf}.
\begin{figure} [H]
\centerline{
\epsfig{file=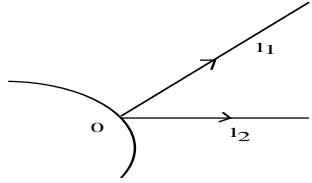,height=3 cm,width=6 cm,angle=0}
}
\caption{Two connected leaves in a graph. }
\label{tleaf}
\end{figure}
\begin{theorem}
A localized eigenvector of the generalized Laplacian exists on two
connected leaves of lengths $l_1,l_2$ if there exists two integers
$p,q$ such that
$$(2 p+1) l_1 - (2 q +1) l_2 =0 .$$
\end{theorem}
\begin{proof}
An eigenvector localized on the two leaves satisfies the
following conditions on the two eigenvector components, \\
$ V_i = A_i \sin k x + B_i \cos k x , ~~~i=1,2$
\begin{eqnarray*}
V_1(0)=V_2(0)=0, \\ 
{V_1}_x(0)+ {V_2}_x(0)=0, \\
{V_1}_x(l_1)=0, \\
{V_2}_x(l_2)=0,
\end{eqnarray*}
From this system of equations we get
\begin{eqnarray*}
B_1 = B_2 =0, \\
A_1+ A_2 =0, \\
\cos k l_1 =0, \\
\cos k l_2 =0,
\end{eqnarray*}
so that
$$k l_1 = (2p +1)  { \pi \over 2 } , ~~~~~k l_2 = (2q +1)  { \pi \over 2 } $$
These conditions are satisfied if $l_1, l_2$ verify
$$(2 p+1) l_1 - (2 q +1) l_2 =0 $$
where $p,q$ are integers.
\end{proof}

Exactly localized eigenvectors also
exist when there are three or more leaves. We sketch the situation for three leaves 
and give the result for $L$ leaves.  \\
Assume there are three leaves. The conditions are then
\begin{eqnarray*}
B_1=B_2=B_3=0,~~A_1+A_2+A_3=0, \\
\cos k l_1 = \cos k l_2=\cos k l_3=0 .
\end{eqnarray*}
From this system we obtain the constraints on the lengths
\begin{eqnarray*}
 (2 p_1 + 1) l_1 - (2 q_1 + 1) l_2 = 0, \\
 (2 p_2 + 1) l_1 - (2 q_2 + 1) l_3 = 0, \\
 (2 p_3 + 1) l_2 - (2 q_3 + 1) l_3 = 0, 
\end{eqnarray*}
where $p_1,p_2,p_3,q_1,q_2,q_3$ are integers.
Note that the eigenspace has dimension 2.

For $L$ connected leaves, we would get an eigenspace of dimension
$L-1$ and $C_L^2$ constraints defining $k$.


%
%
%

\section{Configurations with no exactly localized eigenvectors}

\subsection{No leaf localized eigenvector}

A leaf is an arc such that its end vertex has degree 1, see Fig. \ref{leaf}.
\begin{figure} [H]
\centerline{
\epsfig{file=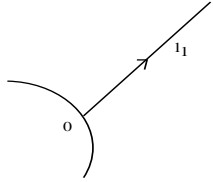,height=3. cm,width=4.05 cm,angle=0}
}
\caption{A leaf in a graph.}
\label{leaf}
\end{figure}
\begin{theorem}
There are no localized eigenvectors on leaves of a metric graph.
\end{theorem}
\begin{proof}
Assume a leaf of length $l$, parameterized by $x \in [0,l]$. 
The boundary conditions at $x=0, l$ are $V(0)=0, ~V_x(0)=V_x(l)=0$.
Writing $V$ as
$$ V = A \sin k x + B \cos k x , $$
where all indices have been dropped for simplicity, we get
from the first two conditions
$$A = B = 0 ,$$
so there are no eigenvectors localized on a leaf.
\end{proof}

\subsection{No localized eigenvector on two connected arcs}

We prove that no localized eigenvector exists on two arcs connected at
one vertex, see Fig. \ref{twoarcs}.
\begin{figure}[H]
\centerline{
\epsfig{file=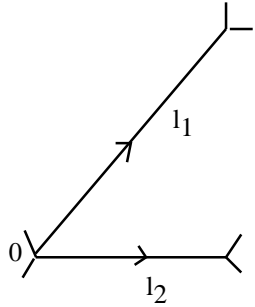,height=4.61 cm,width=4.86 cm,angle=0}
}
\caption{Two connected arcs "embedded" in a graph.}
\label{twoarcs}
\end{figure}
\begin{theorem}
A localized state cannot exist on two arcs connected at one vertex. 
\end{theorem}

\begin{proof}
The localization conditions for the eigenvector components are
\begin{eqnarray*}
V_1(0)=V_2(0)=0, \\
V_1(l_1)=V_2(l_2)=0, \\
V_{1x}(0) + V_{2x}(0)=0, \\
V_{1x}(l_1) + V_{2x}(l_2)=0 , \\
\end{eqnarray*}
which yield the following system
\begin{eqnarray*}
B_1=B_2=0 \\
A_1 s_1=0 \\
A_2 s_2=0 \\
A_1c_1+A_2c_2=0 \\
A_1c_1=0 \\
A_2 c_2=0 ,
\end{eqnarray*}
which only has the solution $A_1=A_2=0$, so a localized eigenvector 
cannot exist.
\end{proof}

\subsection{No degree three vertex eigenvector}

\begin{figure}[H]
\centerline{
\epsfig{file=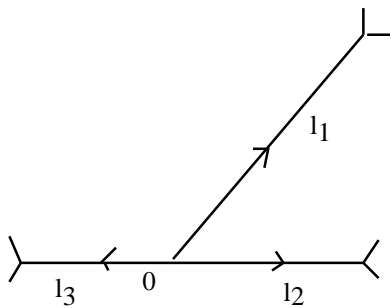,height=4.804 cm,width=6.875 cm,angle=0}
}
\caption{A degree three vertex "embedded" in a graph.}
\label{deg3}
\end{figure}
Consider the configuration of Fig. \ref{deg3} where a degree three
vertex is embedded in a graph, we have the following theorem.
\begin{theorem}
No eigenvector can be localized on three arcs connected at a single
vertex. 
\end{theorem}
\begin{proof}
Let us write the components of an eigenvector localized on the subgraph $l_1,l_2,l_3$.
The conditions are
\begin{eqnarray*}
V_1(l_1)=V_2(l_2)=V_3(l_3)=0, \\ 
{V_1}_x(l_1)= {V_2}_x(l_2)={V_3}_x(l_3)=0 , \\
{V_1}_x(0)+ {V_2}_x(0) + {V_3}_x(0)=0, 
\end{eqnarray*}
yielding
\begin{eqnarray*}
A_1 + A_2+ A_3 =0 , \\
A_1 c_1 - B_1 s_1 =0, \\
A_2 c_2 - B_2 s_2 =0, \\
A_3 c_3 - B_3 s_3 =0, \\
A_1 s_1 + B_1 c_1 =0, \\
A_2 s_2 + B_2 c_2 =0, \\
A_3 s_3 + B_3 c_3 =0, 
\end{eqnarray*}
leading to the homogeneous linear system.
\be \begin{pmatrix} 
1 & . & 1 & . & 1  &  . \cr
c_1 & -s_1 & . & . & .  &  . \cr
s_1 & c_1 & . & . & .  &  . \cr
.   & .   & c_2 & -s_2 & . & .  \cr
.   & .   & s_2 & c_2 & . & .  \cr
.   & .   & .   & .   & c_3 & -s_3 \cr
.   & .   & .   & .   & s_3 & c_3 \cr
\end{pmatrix}
\begin{pmatrix}
A_1 \cr
B_1 \cr
A_2 \cr
B_2 \cr
A_3 \cr
B_3 \cr
\end{pmatrix}
= \begin{pmatrix}
0 \cr
0 \cr
0 \cr
0 \cr
0 \cr
0 \cr
\end{pmatrix}
\ee
The determinant of the submatrix obtained by taking out the first line
is equal to 1, so that the whole matrix has rank greater of equal to 6.
Then there is no other solution than $A_1=B_1=A_2=B_2=A_3=B_3=0$.
\end{proof}

\end{document}